\PassOptionsToPackage{prologue,dvipsnames}{xcolor}
\documentclass[letterpaper]{scrartcl}
\usepackage[total={16.2cm, 23.5cm}]{geometry}
\usepackage{microtype}
\usepackage{amsmath,amsthm,amssymb, mathtools}
\usepackage{tikz}
\usetikzlibrary{patterns,decorations.pathreplacing,arrows.meta,shapes.geometric}
\usepackage{hyperref}
\hypersetup{
pdfencoding=auto, psdextra,
colorlinks=true,citecolor=green!50!black,linkcolor=red!60!black,
}
\usepackage[sort&compress,nameinlink]{cleveref}
\usepackage{nicefrac}
\usepackage[square,numbers,sort]{natbib}
\usepackage{comment}
\usepackage{paralist}
\usepackage{xcolor}
\usepackage{colortbl}
\usepackage{subcaption}
\usepackage{booktabs}
\usepackage{listings}
\usepackage{fancyhdr}
\usepackage{microtype}
\usepackage[linesnumbered,ruled,vlined]{algorithm2e}
\usepackage[linecolor=green!70!black, backgroundcolor=green!10, bordercolor=black, textsize=small]{todonotes}
\definecolor{winered}{rgb}{0.6,0.1,0.1}

\newcommand{\myemph}[1]{{\color{winered}\emph{#1}}}
\newtheorem{theorem}{Theorem}
\newtheorem{lemma}{Lemma}

\theoremstyle{definition}
\newtheorem{definition}{Definition}
\newtheorem{example}{Example}
\theoremstyle{remark}

\crefname{theorem}{Theorem}{Theorems}
\crefname{lemma}{Lemma}{Lemmas}
\crefname{proposition}{Proposition}{Propositions}
\crefname{corollary}{Corollary}{Corollaries}
\crefname{definition}{Definition}{Definitions}
\crefname{example}{Example}{Examples}
\crefname{section}{Section}{Sections}
\crefname{subsection}{Section}{Sections}
\numberwithin{equation}{section}
\newcommand{\cost}{\operatorname{cost}}

\newcommand{\Entropy}{\mathcal E}
\newcommand{\Pay}{\Phi}
\newcommand{\Obj}{\mathrm{sc}}
\newcommand{\pos}[1]{\left(#1\right)_{+}}

\allowdisplaybreaks[1]
\title{Core Up-To-One for Participatory Budgeting with Additive Utilities}

\author{Piotr Skowron}
\date{}
\begin{document}
\maketitle
\thispagestyle{plain}

\begin{abstract}
We study participatory budgeting with additive utilities. We introduce
a selection rule based on utility-weighted harmonic entropy and a cost
penalty, and prove that it returns budget-feasible outcomes satisfying
core-up-to-one.
\end{abstract}

\section{Introduction}\label{sec:introduction}

We extend the notion of harmonic entropy from the recent work of
\citet{BGP2026} on the core in approval-based committee elections
to additive utilities and unequal candidate costs. Our construction
and proof largely follow the ideas of that paper.

\section{Model}\label{sec:model}
Let $C=\{c_1,\ldots,c_m\}$ and $N=\{1,\ldots,n\}$ be the sets of candidates
and voters, respectively, where $m,n\ge1$.
Each candidate $c\in C$ is associated with a cost, denoted by $\cost(c)$,
and there is a budget $b$ that can be spent on selecting candidates.
Consequently, a subset $W\subseteq C$ is \emph{*feasible*} if
$\cost(W)\le b$, where

\begin{align*}
 \cost(W)=\sum_{c\in W}\cost(c).
\end{align*}

Each voter $i$ has a utility $u_i(c) \in [0,1]$ for candidate $c$.
We assume that for each voter $i$ there is at least one candidate $c$ with
$u_i(c)>0$. We extend the notation to sets by writing

\begin{align*}
 u_i(W) = \sum_{c\in W}u_i(c) \qquad(W\subseteq C).
\end{align*}

Thus utilities are nonnegative and additive. Normalizing utilities is without
loss of generality, since the considered notion of core does not contain
comparisons of utilities between different voters.
We write $W+c$ and $W-c$ for $W\cup\{c\}$ and $W\setminus\{c\}$,
respectively.
Below we define our main notion of proportionality.

\begin{definition}[Core]\label{def:core}
A feasible outcome $W$ is in the \emph{*core*} if, for every nonempty
$S\subseteq N$ and every subset of candidate $T\subseteq C$ with $\cost(T)\le b \cdot \nicefrac{|S|}{n}$,
there exists a voter $i\in S$ such that
\begin{align*}
 u_i(T)\le u_i(W).
\end{align*} \hfill $\lrcorner$
\end{definition}

\begin{definition}[Additive-one core]\label{def:additive-core}
A feasible outcome $W$ satisfies the \emph{*core up-to-one*} if, for every
nonempty $S\subseteq N$ and every subset of candidate $T\subseteq C$ with $\cost(T)\le b \cdot \nicefrac{|S|}{n}$, there exists
$i\in S$ such that
\begin{align*}
 u_i(T) < u_i(W) + 1.
\end{align*} \hfill $\lrcorner$
\end{definition}

\section{The Max-Payment-Entropy Rule}\label{sec:rule}
For each candidate $c$, we define $q_c = n \cdot\nicefrac{\cost(c)}{b}$ as the required number of voters that can afford buying candidate $c$.
Given a set of candidates $W\subseteq C$, a \myemph{loose payment system} consists of the set of non-negative payments $\{p_{i}(c)\}_{i\in N, c \in W}$ and remaining budgets of the voters $\{r_{i}\}_{i\in N}$ that satisfy the following consditions:
\begin{enumerate}
  \item all voters have initially one unit of money: $r_i + \sum_{c \in W} p_{i}(c) = 1$ for all $i \in N$,
  \item the voters did not spent more than they initially had: $r_i \ge 0$ for all $i \in N$,
  \item the payments towards selected candidates sum up at most to the required values needed to warrant candidates' selection: $\sum_{i \in N}p_{i}(c) \le q_c$ for all $c \in W$. Note that this is difefrent from a classic axiom of priceability, where the payments are required to sum up exactly to the required values $q_c$. 
  \item the voters can pay only for the candidates they support: $p_i(c) = 0$ for all $i \in N, c \in W$ such that $u_i(c) = 0$.
\end{enumerate}
Let $\Pay(W)$ denote the set of such loose payment systems.

\subsection{Utility-Weighted Harmonic Entropy}

Consider nonnegative masses $x = (x_1, \ldots, x_d)$ with $\sum_{j} x_j=1$, and
positive coordinate weights $a = (a_1, \ldots, a_d)$. The ratio $\nicefrac{x_j}{a_j}$ is the
\myemph{density} of coordinate $j$. For every real value $t \ge 0$, we define

\begin{equation}\label{eq:water-level}
 f_t^a(x) = \max_{J \subseteq \{1, \ldots, d \}} \frac{\sum_{j \in J} x_j}{t + \sum_{j \in J} a_j}.
\end{equation}

In particular:
\begin{align*}
 f_0^a(x)=\max_j\frac{x_j}{a_j},
\end{align*}

\begin{definition}[Weighted harmonic entropy]\label{def:entropy}
The weighted harmonic entropy of a pair of vectors $(x,a)$ is
\begin{equation}\label{eq:entropy}
 F_a(x) = \sum_{\ell = 0}^{\infty} \left(\frac{1}{\ell + 1} - f_\ell^a(x) \right).
\end{equation}
\hfill $\lrcorner$
\end{definition}

Given a payment system $(p = \{p_{i}(c)\}_{i\in N, c \in W}, r = \{r_{i}\}_{i\in N})$ for each voter $i \in N$ we construct the following vectors:
\begin{equation}\label{eq:voter-vectors}
 \begin{aligned}
 x_i &= \bigl(r_i, (p_{i}(c))_{c\in W\colon u_i(c) > 0}\bigr),\\
 a_i &= \bigl(1, (u_i(c))_{c\in W\colon u_i(c) > 0}\bigr).
 \end{aligned}
\end{equation} 
The weighted harmonic entropy of voter's $i$ payments is then given by $F_{a_i}(x_i)$. It is easy to note that $F_{a_i}(x_i)$ is finite, so the entropy is well defined. 

\begin{example}[Prices per unit of utility]
If a voter $i$ values two selected candidates at $1$ and $1/2$, respectively, then the weights are given as
\begin{align*}
  a_i = (1, 1, \nicefrac{1}{2}).
\end{align*}
Assuming the voter pays $\nicefrac{2}{5}$ for the first candidate and $\nicefrac{1}{5}$ for the second, we have:
\begin{align*}
  x_i = (\nicefrac{2}{5}, \nicefrac{2}{5}, \nicefrac{1}{5}).
\end{align*}
\hfill $\lrcorner$
\end{example}

\subsection{The Objective Function}
The weighted harmonic entropy of a subset of candidates $W \subseteq C$ is defined as
\begin{equation}\label{eq:value}
 \Entropy(W) = \max_{(p, r) \in \Pay(W)} \sum_{i\in N} F_{a_i}(x_i).
\end{equation}

The score of a subset $W$ is defined as:
\begin{equation}\label{eq:objective}
 \Obj(W) = \Entropy(W) - \nicefrac{n}{b} \cdot \cost(W).
\end{equation}
\begin{definition}[The Max-Payment-Entropy Rule]\label{def:rule}
The rule selects the sets $W\subseteq C$ that lexicographically maximize $\Obj(W)$. If there are ties, then among the tied outcomes we pick the one with the maximal cost $\cost(W)$.
\end{definition}
Our main theorem is the following.
\begin{theorem}[Main result]\label{thm:main}
The Max-Payment-Entropy Rule return outcomes that satisfy core-up-to-one.
\end{theorem}
The following part of the paper shows several observations that eventually lead to a proof of the theorem.

\section{Proof of the main theorem}\label{sec:proof}

The proof is divided into a sequence of observations. 

\subsection{Observations About Entropy}

In this section we fix a vector $x = (x_1, \ldots, x_d)$ and the accompaning vector of weights $a = (a_1, \ldots, a_d)$.
In order to reduce the notation overhead we will write $f_t(x)$ and $F(x)$ instead of $f_t^a(x)$ and $F_a(x)$, respectively.

Let us define the excess of mass above density $z\ge0$ by
\begin{equation}\label{eq:proof-excess}
 E_x(z)=\sum_{j=1}^d\pos{x_j-a_jz}.
\end{equation}
Note that:
\begin{align*}\label{eq:proof-excess}
 E_x(z)=\sum_{j=1}^d\pos{x_j-a_jz} = \max_{J \subseteq \{1, \ldots, d \}} \sum_{j\in J}^d\pos{x_j-a_jz}.
\end{align*}

First, we observe that $f_t(x)$ is the unique positive solution of
\begin{equation}\label{eq:proof-root}
 E_x(z) = tz.
\end{equation}

To verify this, set $\tau=f_t(x)$. By definition, for each $J \subseteq \{1, \ldots, d \}$ we have
\begin{align*}
\tau \geq \frac{\sum_{j \in J} x_j}{t + \sum_{j \in J} a_j}.
\end{align*}
Also, the above inequality becomeas an equality for the set $J$ that maximizes the right-hand side expression, call it $J^*$.
After reformulation, we get that for each $J \subseteq \{1, \ldots, d \}$:
\begin{align*}
t\tau \geq \sum_{j\in J}(x_j-a_j\tau) \text{.}
\end{align*}
The above inequlity holds for each $J$, so in particular we have that:
\begin{align*}
t\tau \geq E_x(\tau)\text{.}
\end{align*}
Also, we have:
\begin{align*}
t\tau = \sum_{j\in J^*}(x_j-a_j\tau) \leq E_x(\tau)\text{.}
\end{align*}
The two inequalities can be combined into the equality. It is straightforward to observe that the solution to this inequality is unique. 

Moreover, note that $f_t(x)$ is monotone with respect to $t$.

For $s > 0$, we define a transformation $\mathcal T_s$ as follows. Let $\tau = f_s(x)$, we replace the mass of each coordinate $j$ by $\min\{x_j, a_j\tau\}$. Additionally, we introduce one new coordinate with weight $s$ and mass equal to $s\tau$. Equation~\eqref{eq:proof-root} shows that the removed mass is exactly $s\tau$, so the transformed vector again has total mass of one. 

\begin{lemma}\label{lem:proof-water-filling}
For every $s > 0$ and $t \ge 0$,
\begin{equation}\label{eq:proof-shift}
 f_t(\mathcal T_s(x)) = f_{t+s}(x).
\end{equation}
Furthermore, for every $s \in [0, 1]$ it holds that
\begin{align}
 F(\mathcal T_s(x)) - F(x) &\ge s f_0(x),\label{eq:proof-add-potential}\\
 F(\mathcal T_1(x)) - F(\mathcal T_{1-s}(x))  &\le s f_0(x).\label{eq:proof-delete-potential}
\end{align}
\end{lemma}

\begin{proof}
Let $\tau = f_s(x)$. Fix $z \in [0, \tau]$, and consider how applying the operator $\mathcal T_s$ to $x$ changed the value of the excess of mass, that is we are going to bound the expression
\begin{align*}
 E_{\mathcal T_s(x)}(z) - E_x(z) \text{.}
\end{align*}
Obseerve that truncating the existing coordinates at density $\tau$ reduces their excess mass by
runcating the existing coordinates at density $\tau$ reduces their excess mass by
\begin{align*}
\sum_{j=1}^d (x_j - \min\{x_j, a_j\tau\}) = \sum_{j=1}^d\pos{x_j-a_j\tau} = E_x(\tau) = s\tau \text{.}
\end{align*}
The last equality follows from \eqref{eq:proof-root}. Thus, we get that
\begin{align*}
 E_{\mathcal T_s(x)}(z) &= \sum_{j=1}^d\pos{\min\{x_j, a_j\tau\} - a_jz} + \pos{s\tau - sz} \\
                        &= \sum_{j=1}^d\pos{\min\{x_j, a_j\tau\} - a_jz} + s\tau - sz \\
                        &= \sum_{j=1}^d\pos{\min\{x_j, a_j\tau\} - a_jz} + \sum_{j=1}^d (x_j - \min\{x_j, a_j\tau\}) - sz \\
                        &= \sum_{j=1}^d\pos{x_j - a_jz} - sz = E_{x}(z) - sz \text{.}
\end{align*}
For $t > 0$, let $z = f_{t+s}(x) \le \tau$. By \eqref{eq:proof-root}, $E_x(z) = (t+s)z$, so by the above calculations we get that $E_{\mathcal T_s(x)}(z) = E_x(z) - sz = tz$. Here we once again use the fact that $f_t(\mathcal T_s(x))$ is the unique solution of the above inequality (cf., \eqref{eq:proof-root}), from which we get $f_t(\mathcal T_s(x)) = z = f_{t+s}(x)$. This completes the proof of \eqref{eq:proof-shift}.

The function $t\mapsto f_t(x)$ is decreasing and convex: it is a maximum of functions $B / (t+A)$ with $B \ge 0$ and $A > 0$. Fix $\ell \ge 0, s \in [0, 1]$, and observe that 
\begin{align*}
  l + s = s(\ell + 1) + (1 - s)\ell \text{.}
\end{align*}
Hence
\begin{align*}
 f_{\ell + s}(x) \le (1 - s)f_\ell(x) + sf_{\ell+1}(x) \qquad \text{for all~} \ell \ge 0, s\in[0, 1].
\end{align*}
Using \eqref{eq:proof-shift} and telescoping yields
\begin{align*}
 F(\mathcal T_s(x)) - F(x) &= \sum_{\ell = 0}^{\infty} \bigl(f_\ell(x) - f_{\ell} (\mathcal T_s(x))\bigr)\\
                           &= \sum_{\ell = 0}^{\infty} \bigl(f_\ell(x) - f_{\ell+s}(x)\bigr)\\
                           &\geq  \sum_{\ell = 0}^{\infty} \bigl(f_\ell(x) - (1 - s)f_\ell(x) - sf_{\ell+1}(x)\bigr)\\
                           &= \sum_{\ell = 0}^{\infty}\bigl(f_\ell(x)-f_{\ell+1}(x)\bigr) = s f_0(x),
\end{align*}
This proves \eqref{eq:proof-add-potential}. Next, we take $\ell \ge 0, s \in [0, 1]$, and observe that 
\begin{align*}
  \ell + 1 - s = s\ell + (1 - s)(\ell + 1) \text{.}
\end{align*}
Thus, similarly as before, convexity gives
\begin{align*}
 f_{\ell+1-s}(x)\le s f_\ell(x) + (1-s)f_{\ell + 1}(x) \qquad \text{for all~} \ell \ge 0, s\in[0, 1].
\end{align*}
Therefore
\begin{align*}
 F(\mathcal T_1(x)) - F(\mathcal T_{1-s}(x)) &= \sum_{\ell = 0}^{\infty} \bigl(f_{\ell}(\mathcal T_{1-s}(x)) - f_{\ell}(\mathcal T_1(x))\bigr)\\
                                             &= \sum_{\ell = 0}^{\infty} \bigl(f_{\ell + 1 - s}(x) - f_{\ell + 1}(x)\bigr)\\
                                             &\leq \sum_{\ell = 0}^{\infty} \bigl(s f_\ell(x) + (1-s)f_{\ell + 1}(x)- f_{\ell + 1}(x)\bigr)\\
                                            &\le s \sum_{\ell = 0}^{\infty} \bigl(f_\ell(x)-f_{\ell+1}(x)\bigr) = s f_0(x).
\end{align*}
This proves \eqref{eq:proof-delete-potential}.
\end{proof}

We will also use an observation, that two coordinates of equal density
can be merged, adding their masses and weights, without changing the
entropy. Indeed if their common density is $h$, their combined contribution
to \eqref{eq:proof-excess} is
\begin{align*}
 \pos{a_j(h-z)} + \pos{a_{j'}(h-z)} = \pos{(a_j+a_{j'})(h-z)},
\end{align*}
which is exactly the contribution of the merged coordinates. Thus, the solution to \eqref{eq:proof-root} does not change.

\subsection{Balanced optimal payments}

Next we will construct payment systems that witness certain bounds on the entropy.

\begin{lemma}\label{lem:proof-balanced}
For each subset of candidates $W\subseteq C$, there is an entropy-maximizing loose payment system satisfying
\begin{equation}\label{eq:proof-balanced}
 p_i(c)\le r_i u_i(c) \qquad \text{for all~~} i\in N, c\in W \text{,}
\end{equation}
and
\begin{equation}\label{eq:proof-slack-balance}
 \sum_{i\in N} p_i(c) < q_c \quad \text{only if~~} p_i(c) = r_i u_i(c) \text{~~for all~~} i\in N.
\end{equation}
For such payments,
\begin{equation}\label{eq:proof-reserve}
 f_0(x_i) = r_i \ge \frac{1}{1 + u_i(W)}.
\end{equation}
\end{lemma}

\begin{proof}
First consider an arbitrary unit-mass vector. Observe what happens when we transfer a sufficiently
small amount of mass from a coordinate of higher density to a
coordinate of lower density, without letting their densities cross.
For each value of $z$, either both coordinates contribute to $E_x(z)$,
neither contributes, or only the higher-density coordinate contributes.
During the transfer the excess mass therefore remains constant or
decreases. Thus, every $f_t$, including the maximum density $f_0$,
weakly decreases, and $F$ weakly increases. 
The same transfer strictly decreases the following quantity $\sum_{j=1}^{d} \nicefrac{x_j^2}{a_j}$ (it is easy to see that by writing down the derivative of the transfer).

Among the payment systems attaining $\Entropy(W)$, choose one
minimizing
\begin{equation}\label{eq:proof-balance-quadratic}
 \sum_{i\in N}\left(r_i^2+
 \sum_{c\in W\colon u_i(c) > 0}\frac{p_i(c)^2}{u_i(c)}\right).
\end{equation}
If $p_i(c)/u_i(c)>r_i$ for a positively valued candidate, a small transfer
from its payment to the reserve preserves all payment constraints,
weakly increases entropy, and strictly decreases
\eqref{eq:proof-balance-quadratic}. This contradicts the choice of
payments and proves \eqref{eq:proof-balanced}.

If a candidate $c$ is underfunded and $p_i(c)/u_i(c)<r_i$ for a voter
valuing $c$, a sufficiently small transfer in the opposite direction
is feasible and gives the same contradiction. Hence
\eqref{eq:proof-slack-balance} holds.

The reserve has weight one, so \eqref{eq:proof-balanced} gives
$f_0(x_i) = r_i$. Finally, again by \eqref{eq:proof-balanced} we get:
\begin{align*}
 1 = r_i + \sum_{c\in W}p_i(c) \le r_i \bigl(1 + u_i(W)\bigr) \text{,}
\end{align*}
which completes the proof
\end{proof}

We call payments satisfying this lemma \emph{balanced}.

\subsection{Change of the Entropy due to Adding a Candidate}

\begin{lemma}\label{lem:proof-addition}
Let $(p, r)$ be a balanced optimal payment system for $W\subseteq C$, and let $c \notin W$. Let
\begin{equation}\label{eq:proof-reserve-load}
 R_c = \sum_{i \in N}r_i u_i(c) \text{.}
\end{equation}
Then
\begin{equation}\label{eq:proof-addition-bound}
 \Entropy(W + c) - \Entropy(W) \ge \min\{R_c, q_c\}.
\end{equation}
If $R_c > q_c$, then the left-hand side is strictly greater than $q_c$.
\end{lemma}

\begin{proof}
For each voter with $u_i(c)>0$, apply $\mathcal T_{u_i(c)}$ to $x_i$,
interpreting the added coordinate as their payment to $c$. Retain the
payment vectors of voters with $u_i(c)=0$. By
\eqref{eq:proof-add-potential} and \eqref{eq:proof-reserve}, the total
entropy gain $G$ of this system satisfies
\begin{equation}\label{eq:proof-raw-gain}
 G = \sum_{i \in N} F(T_{u_i(c)}(x_i)) - F(x) \geq \sum_{i \in N} u_i(c)f_0(x_i) = \sum_{i \in N} u_i(c)r_i =  R_c.
\end{equation}

By the definition of the transformation, the total new payment to $c$ is
\begin{equation}\label{eq:proof-raw-payment}
 P^+ = \sum_{i\colon u_i(c) > 0}u_i(c)f_{u_i(c)}(x_i) \leq \sum_{i\colon u_i(c) > 0} u_i(c) f_{0}(x_i) \leq R_c.
\end{equation}

If $R_c > 0$, at least one $u_i(c)$ is strictly positive, and so at least one inequality $f_{u_i(c)}(x_i) \leq f_{0}(x_i)$ is strict.
As a result, this gives the strict inequality $P^+ < R_c$.

All existing candidate payments weakly decrease, so only the cap of
$c$ might be violated. If $P^+ \le q_c$, the new system is feasible for $W + c$ and gives
\begin{align*}
 \Entropy(W + c) - \Entropy(W) \ge G \ge R_c.
\end{align*}
This proves both conclusions in this case. 

If $P^+ > q_c$, take the
convex combination of the new system and the original system with weights equal to $\theta = q_c / P^+$ and $1 - \theta$, respectively. All payment constraints are then satisfied. We get:
\begin{align*}
 \Entropy(W + c) - \Entropy(W) \geq \theta G \geq \frac{q_c R_c}{P^+} > q_c \text{.}
\end{align*}
The inequality is strict because $P^+ < R_c$. This completes the proof.
\end{proof}

\subsection{Loss of Entropy due to Deleting a Candidate}

\begin{lemma}\label{lem:proof-deletion}
If a set $W\subseteq C$ satisfies $\sum_{c\in W}q_c>n$, then there
exists $d\in W$ such that
\begin{equation}\label{eq:proof-deletion-bound}
 \Entropy(W)-\Entropy(W-d)<q_d.
\end{equation}
\end{lemma}

\begin{proof}
Take balanced optimal payments for $W$, and write
\begin{align*}
 P_c = \sum_{i \in N}p_i(c), \qquad U = \{c \in W\colon P_c < q_c \}.
\end{align*}
We redistribute some reserves while respecting the caps, choose a
candidate that is still underfunded, and then delete it with entropy
loss bounded by its redistributed payments.

\paragraph{Redistributing Reserves.}
For $D\subseteq U$, let
\begin{align*}
 N(D) = \{i \in N\colon \text{there exists~} c \in D \text{~with~} u_i(c) > 0\}
\end{align*}
and set
\begin{equation}\label{eq:proof-gamma}
 \Gamma(D) = \sum_{c \in D}(q_c - P_c) - \sum_{i \in N(D)}r_i \text{.}
\end{equation}
Every candidate outside $U$ receives its full cap. Therefore
\begin{align*}
 \Gamma(U) = \sum_{c\in W}q_c - \sum_{c \in W}P_c - \sum_{i \in N(U)}r_i = \sum_{c \in W}q_c - n + \sum_{i \notin N(U)}r_i > 0.
\end{align*}

Choose $D\subseteq U$ maximizing $\Gamma(D)$.  For every $R\subseteq D$, maximality gives $\Gamma(D) \ge \Gamma(D\setminus R)$, or equivalently
\begin{align*}
  \sum_{c \in D}(q_c - P_c) - \sum_{i \in N(D)}r_i \geq \sum_{c \in D\setminus R}(q_c - P_c) - \sum_{i \in N(D\setminus R)}r_i
\end{align*}
After reformulation:
\begin{equation}\label{eq:proof-flow-condition}
 \sum_{c \in R}(q_c - P_c) \ge \sum_{i \in N(D)\setminus N(D\setminus R)}r_i.
\end{equation}

We will now show that every voter in $N(D)$ can distribute
their entire reserve among positively valued candidates in $D$,
without exceeding the remaining candidate capacities. To prove this,
construct a flow network whose source is connected to each
$i\in N(D)$ by an edge of capacity $r_i$. Connect $i$ to $c\in D$
by an unlimited-capacity edge precisely when $u_i(c)>0$, and connect
$c$ to the sink by an edge of capacity $q_c-P_c$.

Consider a cut of finite capacity. Let $R$ be its source-side
candidates and $I$ its source-side voters. No unlimited-capacity edge
can cross the cut, so if $i \in I$ and $c \in D \setminus R$ then $u_{i}(c) = 0$.
In particular, $i \notin N(D\setminus R)$, and so
$I\subseteq N(D)\setminus N(D\setminus R)$. Its capacity is therefore
\begin{align*}
 \sum_{i\in N(D)\setminus I}r_i + \sum_{c\in R}(q_c-P_c) &\geq \sum_{i \in N(D) \setminus I}r_i + \sum_{i \in N(D)\setminus N(D\setminus R)}r_i \\
                                                        & \geq \sum_{i \in N(D) \setminus I}r_i + \sum_{i\in I}r_i
 = \sum_{i\in N(D)}r_i,
\end{align*}
where the second inequality uses \eqref{eq:proof-flow-condition}. The
max-flow/min-cut theorem gives a flow saturating all source--voter
edges, as required.

Write $\widehat p_i(c)$ for the payments after this redistribution.
Payments outside $D$ are unchanged. Within $D$, balancedness and
$D\subseteq U$ give
\begin{equation}\label{eq:proof-redistributed-lower}
 \widehat p_i(c)\ge p_i(c) = r_i u_i(c)\qquad \text{~for all~} i\in N, c\in D.
\end{equation}
Only positively valued candidates receive additional payments, and
\begin{equation}\label{eq:proof-redistributed-budget}
 \sum_{c\in D}\widehat p_i(c) =\sum_{c\in D} p_i(c) +r_i \qquad \text{~for all~} i\in N(D).
\end{equation}
All candidate caps are respected. Moreover,
\begin{equation}\label{eq:proof-remaining-slack}
 \sum_{c\in D}\left(q_c - \sum_{i \in N}\widehat p_i(c)\right) =\Gamma(D)>0.
\end{equation}
Thus some candidate remains underfunded after redistribution.

\paragraph{Choosing the Candidate to Delete.}
Among all such redistributions, choose one minimizing
\begin{equation}\label{eq:proof-redistribution-quadratic}
 \sum_{i\in N}\sum_{c\in D\colon u_i(c) > 0} \frac{\widehat p_i(c)^2}{u_i(c)}.
\end{equation}
By \eqref{eq:proof-remaining-slack}, we there exists $d\in D$ with
\begin{equation}\label{eq:proof-chosen-candidate}
 \sum_i\widehat p_i(d)<q_d.
\end{equation}

For every voter $i$ with $u_i(d)>0$ it must hold that:
\begin{equation}\label{eq:proof-maximum-density}
 \frac{\widehat p_i(d)}{u_i(d)} \ge\frac{\widehat p_i(c)}{u_i(c)} \qquad \text{~for all~} c\in D\colon u_i(c)>0.
\end{equation}
Otherwise, a candidate $c$ of larger density would have a payment
strictly above its lower bound $r_i u_i(c)$, because the density of
$d$ is at least $r_i$ (because it was initially at least $r_i$, and by transfers it could only increase). Transferring a sufficiently small amount from
$c$ to $d$ preserves the lower bounds and the voter's total payment.
It also preserves the caps, since $d$ has slack. Its derivative in
\eqref{eq:proof-redistribution-quadratic} is negative, contradicting
minimality. This proves \eqref{eq:proof-maximum-density}.

Outside $D$, payments are unchanged and their densities are at most
$r_i$ by \eqref{eq:proof-balanced}. Consequently, for every voter
valuing $d$, its redistributed coordinate has maximum density among
all of that voter's payments.

\paragraph{Replacing the Deleted Coordinate by a Reserve.}
Let us fix a voter $i$ with $u_i(d)>0$, and let $y_i$ be the vector of their
redistributed payments with no reserve. This vector has
total mass one. Its excess
mass at density $r_i$ is
\begin{equation}\label{eq:proof-recover-level}
 E_{y_i}(r_i) =\sum_{c\in D}\bigl(\widehat p_i(c) - r_i u_i(c)\bigr) = r_i.
\end{equation}
This is because for candidates in $c \in D$ their original payments were equal to $p_i(c) = r_i u_i(c)$, and the total value of $r_i$ was redistributed. 

From the above equality, and by \eqref{eq:proof-root} we get that $f_1(y_i)=r_i$. The operator $T_1$ thus redistributes the mass above the density of $r_i$.  
As a result, we get that applying $T_1$ to the vector $y_i$ simply recovers the initial vector $x_i$:
\begin{equation}\label{eq:proof-recover-vector}
 x_i=\mathcal T_1(y_i).
\end{equation}

Let $s = u_i(d) \in (0,1]$. If $s < 1$, we apply
$\mathcal T_{1-s}$ to $y_i$ and let $\tau = f_{1-s}(y_i)$. Since $d$
has maximum density, its truncated coordinate (after applying $\mathcal T_{1-s}$) has weight $s$ and
mass $s\tau$. The new coordinate has weight $1-s$ and mass
$(1-s)\tau$. We merge these two coordinates of density $\tau$ into a
single coordinate of weight one and mass $\tau$, and we use that
coordinate as the new reserve. This removes candidate $d$.
If $s=1$, we simply get $d$ as the new reserve.

Let us call the resulting payment vector $x_i^-$. Equal-density merging
preserves entropy, so in both cases
\begin{align*}
 F(x_i^-)=F(\mathcal T_{1-s}(y_i)).
\end{align*}
By \eqref{eq:proof-delete-potential},
\eqref{eq:proof-recover-vector}, and the maximum-density property,
\begin{align}
 F(x_i) - F(x_i^-) = F(\mathcal T_1(y_i))-F(\mathcal T_{1-s}(y_i)) \le s f_0(y_i) =s\,\frac{\widehat p_i(d)}{u_i(d)} = \widehat p_i(d).\label{eq:proof-voter-deletion-loss}
\end{align}

For voters with $u_i(d)=0$, we keep their original payment vectors. This way we constructed a loose payment system for $W-d$.

\paragraph{Asseasing the Total Loss.}

The entropy of every voter not valuing $d$ is unchanged. Therefore
\begin{align*}
 \Entropy(W)-\Entropy(W-d) \le \sum_{i\colon u_i(d) > 0} \bigl(F(x_i) - F(x_i^-)\bigr) \le \sum_i\widehat p_i(d) < q_d,
\end{align*}
using \eqref{eq:proof-voter-deletion-loss} and
\eqref{eq:proof-chosen-candidate}.
This proves the lemma.
\end{proof}

\subsection{Completing the Proof}

\begin{proof}[Proof of \cref{thm:main}]
Let $W$ be any set selected by the rule. First, $W$ is feasible. Indeed, for the sake of contradiction assume that
\begin{align*}
 \sum_{c\in W}q_c=\frac nb\cost(W)>n,
\end{align*}
By \cref{lem:proof-deletion} there exists $d\in W$ such that
$\Entropy(W) - \Entropy(W - d) < q_d$. It follows that
\begin{align*}
 \Obj(W - d) - \Obj(W) =q_d - \bigl(\Entropy(W)-\Entropy(W - d)\bigr) > 0,
\end{align*}
contradicting the maximality of $\Obj(W)$.

Now, choose balanced optimal payments for $W$ as in
\cref{lem:proof-balanced}. We claim that
\begin{equation}\label{eq:proof-losing-certificate}
 \sum_{i \in N}r_i u_i(c) < q_c \qquad \text{for all~} c \notin W.
\end{equation}

Indeed, fix a candidate $c \notin W$, and let us denote the left-hand side by $R_c$.
If we would have that $R_c > q_c$, then \cref{lem:proof-addition} would imply
\begin{align*}
 \Obj(W + c) - \Obj(W) = \Entropy(W + c) - \Entropy(W) -q_c > 0,
\end{align*} 
a contradiction.

If $R_c = q_c$, the same lemma yields
$\Obj(W + c) \ge \Obj(W)$. Since in the case of equal scores we pick the solution with the higher cost, we get once again a contradiction. 

Now let $S\subseteq N$ be nonempty and let $T\subseteq C$ satisfy
\begin{equation}\label{eq:proof-affordable}
 \cost(T) \le b\frac{|S|}{n}.
\end{equation}
Suppose, for a contradiction, that
\begin{equation}\label{eq:proof-blocking}
 u_i(T) \ge u_i(W) + 1 \qquad \text{for all~} i \in S.
\end{equation}
Define personalized prices for every candidate by
\begin{align*}
 \pi_i(c)=
 \begin{cases}
 p_i(c), & c \in W,\\
 r_i u_i(c), & c \notin W,
 \end{cases}
\end{align*}

For each $i\in S$, the unit-budget identity and
\eqref{eq:proof-balanced} give
\begin{align*}
 \pi_i(T) =\sum_{c\in T}\pi_i(c) 
 &=\sum_{c\in T\cap W}p_i(c)+r_i u_i(T\setminus W)\\
 &=1-r_i-\sum_{c\in W\setminus T}p_i(c)
      +r_i u_i(T\setminus W)\\
 &\ge1-r_i-r_i u_i(W\setminus T)+r_i u_i(T\setminus W)\\
 &=1-r_i+r_i\bigl(u_i(T)-u_i(W)\bigr)
 \ge1,
\end{align*}
where the last inequality uses \eqref{eq:proof-blocking}.
Thus $\sum_{i\in S}\pi_i(T)\ge|S|$.

On the other hand, $T\setminus W\ne\varnothing$, since utilities
are nonnegative and \eqref{eq:proof-blocking} would be impossible
for $T\subseteq W$. A selected candidate's personalized prices sum
to at most $q_c$ by the loose payment constraints, and an unselected
candidate's prices sum to strictly less than $q_c$ by
\eqref{eq:proof-losing-certificate}. All prices are nonnegative.
Consequently,
\[
 \sum_{i\in S}\pi_i(T)
 \le\sum_{i\in N}\pi_i(T)
 <\sum_{c\in T}q_c
 =\frac nb\cost(T)
 \le|S|,
\]
where the final inequality follows from \eqref{eq:proof-affordable}.
This is a contradiction. Hence some voter $i\in S$ satisfies
$u_i(T)<u_i(W)+1$. Together with the feasibility of $W$, this is
precisely core-up-to-one as stated in \cref{def:additive-core}.
\end{proof}

\end{document}